\documentclass{article}[12pt]
\usepackage[utf8]{inputenc}
\usepackage{amsmath, amssymb, amsthm}
\usepackage{mathtools}
\usepackage{amsfonts}
\usepackage{commath}
\usepackage[margin=1.15in]{geometry} 
\title{\textbf{Streaming Hardness of Unique Games}}

\author{Venkatesan Guruswami\thanks{Computer Science Department, Carnegie Mellon University, 5000 Forbes Ave, Pittsburgh, PA, USA, 15213. Email: \texttt{venkatg@cs.cmu.edu}. Research supported in part by NSF grants CCF-1422045 and CCF-1526092.} 
\and Runzhou Tao\thanks{Institute for Interdisciplinary Information Sciences, Tsinghua University, Beijing, China 100084. Email: \texttt{trz15@mails.tsinghua.edu.cn}. Most of this work was done during a visit by the author to Carnegie Mellon University.}}
\date{}

\newtheorem{theorem}{Theorem}[section]

\newtheorem{lemma}[theorem]{Lemma}
\newtheorem{definition}[theorem]{Definition}

\begin{document}

\maketitle
\thispagestyle{empty}

\begin{abstract}
    We study the problem of approximating the value of a Unique Game instance in the streaming model. A simple count of the number of constraints divided by $p$, the alphabet size of the Unique Game, gives a trivial $p$-approximation that can be computed in $O(\log n)$ space. Meanwhile, with high probability, a sample of $\tilde{O}(n)$ constraints suffices to estimate the optimal value to $(1+\epsilon)$ accuracy. We prove that any single-pass streaming algorithm that achieves a $(p-\epsilon)$-approximation requires $\Omega_\epsilon(\sqrt{n})$ space. Our proof is via a reduction from lower bounds for a communication problem that is a $p$-ary variant of the Boolean Hidden Matching problem studied in the literature. 
    Given the utility of Unique Games as a starting point for reduction to other optimization problems, our strong hardness for approximating Unique Games could lead to down\emph{stream} hardness results for streaming approximability for other CSP-like problems.
  
\end{abstract}

\section{Introduction}

The Unique Games (UG) problem is a type of constraint satisfaction problem on a graph. Given an alphabet $[p]=\{0,1,\dots,p-1\}$ and a graph $G = (V, E)$, we need to find a label assignment $x: V \to [p]$. The constraint on an edge $(u,v) \in E$ is specified described by a permutation $\pi_{uv}: [p] \to [p]$ and we want to find the assignment to maximize the number of equations $\pi_{uv}(x_u) = x_v$ that are satisfied. This maximum possible value over all possible assignments is called the optimal value of the UG instance.  Simply picking a random assignments satisfies a fraction $1/p$ of the constraints in expectation, giving a trivial factor $p$ approximation algorithm to the optimal value of any instance. More sophisticated algorithms based on semidefinite programming give better approximation guarantees~\cite{CMM}, but even on almost-satisfiable instances where the optimal value is a $(1-\epsilon)$ fraction of the total number of constraints, the algorithm satisfies only a fraction $\approx p^{-\epsilon/2}$ of the constraints. Under Khot's celebrated Unique Games conjecture~\cite{Khot-UGC}, this guarantee cannot be improved~\cite{KKMO}, and the conjecture further implies optimal hardness results for a host of problems. 
In terms of proven hardness results (under say the standard assumption that $\mathrm{P} \neq \mathrm{NP}$), we know that Unique Games does not admit any constant factor approximations~\cite{feige-reichman}, and in an exciting recent line of work this was also established on instances that have optimum value close to a fraction $1/2$~\cite{DKKMS-stoc18,KMS-focs18}.

To shed further light on the (difficulty of the) Unique Games problem from a different angle, in this work, we consider the Unique Games problem in the streaming model of computation. The constraints are assumed to arrive one-by-one in a single pass. The algorithm is only given a limited amount of memory, so cannot store the entire instance as it passes by. The goal of the algorithm is to estimate the optimal value of the Unique Games instance. That is, it must output a value $T$ which is a lower bound on the optimum number of constraints that can be satisfied, and which is at most an approximation factor $f$ from the optimum. 
In recent years, numerous algorithms and hardness for problems in the streaming model have been developed, and this work address the important Unique Games problem from the streaming perspective.

The simple-minded algorithm which simply counts the number of constraints and outputs a $1/p$ fraction of it as a valid estimate for every instance (by virtue of the random assignment algorithm), and delivers a factor $p$ approximation. This algorithm can obviously be implemented in the streaming model using $O(\log n)$ space. Meanwhile, if we are given $\tilde{O}(n)$ space, we can sample a random $\tilde{O}(n)$-size subset of constraints and the answer of sampled unique game gives us an arbitrarily close approximation for the original stream.\footnote{Note that We do not place any computational restriction on the algorithm, only on the amount of space it may use. Also, since we are talking about sub-linear space, we do not focus on finding an approximate solution, but only outputting an estimate of the optimal value. Since our focus in on lower bounds, this only makes our technical result stronger.}
A natural question which arises, and which motivates this work, is thus: \emph{can we do better than the trivial factor $p$ approximation in polylogarithmic space?}

In a beautiful work, Kapralov, Khanna, and Sudan~\cite{kapralov2015streaming} showed that the problem of Max-CUT, which is a special case of the Unique Games problem with alphabet size $2$, does not admit an approximation better than the trivial factor $2$ in $o(\sqrt{n})$ space in the streaming model where the edges arrive one-by-one. On the other hand, a recent work~\cite{GVV-approx17} showed that for the Max 2CSP problem (arbitrary Boolean arity two constraints) and Max-DICUT (the analog of Max-CUT on directed graphs), one can in fact beat the trivial factor $4$ algorithm (that outputs $1/4$'th the number of constraints, which is the expected value of a random assignment), and achieve a $\approx 5/2$-approximation using $O(\log n)$ space.
The status of the streaming approximability of Unique Games over larger alphabet sizes was not addressed and remained open until our work.

\subsection{Our Result}
We show that for Unique Games with alphabet size $p$, a single-pass streaming algorithm requires at least $\tilde{\Omega}(\sqrt{n})$ space to have any chance of delivering a better estimate than the trivial factor $p$ approximation. In particular, we cannot beat the trivial constraint-counting algorithm in the worst-case in polylogarithmic space.

\begin{theorem}\label{main}
Let $p \ge 2$ be an integer and $\epsilon > 0$ be a small constant. Any streaming algorithm giving $(p-\epsilon)$-approximation for Unique Games with alphabet size $p$ with success probability at least $9/10$ over its internal randomness must use $c_{p,\epsilon} \sqrt{n}$ space, for some positive constant $c_{p,\epsilon}$ that depends only on $p,\epsilon$.

Furthermore, the hardness holds for distinguishing between satisfiable instances and those for which at most a fraction $(1/p+\epsilon)$ of the constraints can be satisfied by any assignment, and when the Unique Games constraints are linear (of the form $x_u + x_v = \alpha_{uv}$ over integers mod $p$).
\end{theorem}

\subsection{Proof Structure}
In our proof, we first introduce in Section~\ref{sec:p-ary-HM}, a communication problem called the $p$-ary Hidden Matching problem, which is a $p$-ary variant of the (Boolean) Hidden Matching problem proposed by Gavinsky et al.\cite{gavinsky2007exponential} and first used for streaming lower bounds by Verbin and Yu in \cite{Verbin2011The}. The (distributional) $p$-ary Hidden Matching problem is a two-party one-way communication problem where Alice holds a random $p$-ary vector $x \in \mathbb{Z}_p^n$ and Bob holds a random matching of size $r = \alpha n$ (for some suitable $\alpha \in (0,1)$) and a vector $w \in \mathbb{Z}_p^r$. 
Alice must send one message to Bob, based on which he must distinguish between two distributions on the inputs. In both distributions $x$ is uniformly random, and $M$ is a random matching of the prescribed size. In the YES distribution, we set $w_e = x_u + x_v$ for each $e = (u,v)$ in the matching (i.e., $w = Mx$ where $M \in \{0,1\}^{\alpha n \times n}$ is the incidence matrix of the matching); in the NO distribution, $w$ is uniformly random. We prove a communication lower bound of this problem using Fourier-analytic methods, which is similar to \cite{Kahn1988The}.

The vector $w$ and the matching in the $p$-ary Hidden Matching problem can be seen as a description of some Unique Game constraints $x_u + x_v = w_e$. Of course each such instance individually is trivially always satisfiable. We can construct hard instances of Unique Game by combining together $O(1/\epsilon^2)$ independent copies of the random matching and corresponding $w$. In the YES case, we let $w$ be according to the \emph{same} (random) $x$, so that the constraints can be satisfied by $x$. In the NO case, the various choices of $w$ are random and independent. This implies that every assignment $x \in \mathbb{Z}_p^n$ is close in performance to a random assignment, and thus satisfies only $\approx 1/p$ of the constraints, by concentration bounds. 

We prove that a low-space streaming algorithm cannot distinguish between these distributions, which then implies Theorem~\ref{main}.
To prove this indistinguishability result, we give a reduction from
the $p$-ary Hidden Matching problem. The proof is a classical hybrid argument since the streaming instance can be seen as a ``multi-stage'' version of the communication problem.

\subsection{Differences from \cite{kapralov2015streaming}}
Our approach heavily borrows from the Max-CUT streaming lower bound from of Kapralov, Khanna, and Sudan~\cite{kapralov2015streaming}. Compared to their work, we only prove Theorem~\ref{main} for a \emph{worst-case} arrival order of constraints, whereas the Max-CUT hardness result is shown even for a \emph{random} arrival order for the edges. At each stage, instead a matching, Kapralov et. al. used a sub-critical random Erd\"{o}s-R\'{e}nyi graph with edge probability $\approx \alpha/n$. If the parameter $\alpha$ is sufficiently small, the graph obtained by putting together edges from all the stages is close in distribution to a random graph. As a result the arrival of edges in a random order does not help the streaming algorithm. For the analysis of each stage, they use a communication problem called the Boolean Hidden Partition problem that is variant of the Boolean Hidden Matching problem, since they have to work with Erd\"{o}s-R\'{e}nyi graphs rather than random matchings. This requires changes to some components in the proof outline of \cite{gavinsky2007exponential,Verbin2011The}. 

Our communication problem still concerns matchings (rather than sub-critical Erd\"{o}s-R\'{e}nyi graphs), though we allow for (components of) $x,w$ to take values from $\mathbb{Z}_p$ instead of Boolean values. By using Fourier analysis over the group $\mathbb{Z}_p$ instead of $\mathbb{Z}_2$, we are able to adapt the communication lower bound of \cite{gavinsky2007exponential}.

It remains an interesting question to prove a streaming hardness for Unique Games similar to Theorem~\ref{main} for the case of random arrival order of constraints. 

\section{Preliminaries}
Let $\mathbb{Z}_p = \{0,1,\dots,p-1\}$ be the ring with addition and multiplication modulo $p$. (We do not assume that $p$ is a prime.)
Fourier analysis over $\mathbb{Z}_p^n$ plays a key role in our proof. Consider the space of functions $\mathbb{Z}_p^n \to \mathbb{C}$. We define the inner product and $2$-norm in it by
$$\langle f, g \rangle = \frac{1}{p^n} \sum_{x \in \mathbb{Z}_p^n} f(x)\overline{g(x)} \qquad \norm{f}_2^2 = \langle f, f \rangle = \frac{1}{p^n} \sum_{x \in \mathbb{Z}_p^n} |f(x)|^2$$
\noindent
The Fourier transform of $f$ is a function $\hat{f}:\mathbb{Z}_p^n \to \mathbb{C}$ defined by
$$\hat{f}(z)=\langle f, \chi_z \rangle = \frac{1}{p^n} \sum_{x \in \mathbb{Z}_p^n} f(x)\overline{\omega^{z \cdot x}}$$
\noindent where $\chi_z : \mathbb{Z}_p^n \to \mathbb{C}$ is the character $\chi_z(x) =\omega^{z \cdot x} $ with ``$\cdot$'' being the scalar product and $\omega = e^{2\pi i/p}$ being the primitive $p$'th root of unity. For $z \in \mathbb{Z}_p^n$, we denote by $|z|$ the number of nonzero entries in $z$.

In our later proof, we use the following two lemmas concerning Parseval's identity and hypercontractivity.
\begin{lemma}[Parseval]
For every function $f : \mathbb{Z}_p^n \to \mathbb{C}$, we have
$$\norm{f}_2^2 = \sum_{z\in \mathbb{Z}_p^n} |\hat{f}(z)|^2.$$
\end{lemma}

\begin{lemma}[Hypercontractivity Theorem, \cite{o2014analysis}]
For function $f \in L_2(\mathbb{Z}_p^n)$, if $1 < q < 2$ and $ 0 \le \rho \le \sqrt{q-1} (1/p)^{1/q-1/2}$, we have
$$\norm{T_\rho f}_2 \le \norm{f}_q$$
where $T_\rho$ is the operator defined by $T_\rho f(x) = \sum_{z \in \mathbb{Z}_p^n} \hat{f}(z) \rho^{|z|} \chi_z(x)$.
\end{lemma}
Using the above theorem, we can derive an estimate on the sum of Fourier coefficients weighted by its support size.

\begin{lemma}\label{lemhct}
For a set $A \subseteq \mathbb{Z}_p^n$ and let $f$ be its indicator function and let $|z|$ denote the number of non-zero coordinates of $z\in \mathbb{Z}_p^n$. Then  for every $\delta \in [0,1/p]$, we have
$$\sum_{z \in \mathbb{Z}_p^n} \delta^{|z|} |\hat{f}(z)|^2 \le  \left( \frac{|A|}{p^n}\right)^{2/(1+p \delta)}.$$
\end{lemma}

\begin{proof}
Let $\rho = \sqrt{q-1}(1/p)^{1/2} \le \sqrt{q-1} (1/p)^{1/q-1/2}$, then $q = 1+p \rho^2$. By the hypercontractivity theorem, we know that
$$\norm{T_\rho f}_2 \le \norm{f}_{1+p \rho^2}$$

Meanwhile, we have $\norm{T_\rho f}_2^2 = \sum_{z \in \mathbb{Z}_p^n} \rho^{2|z|} |\hat{f}(z)|^2$. Taking the square of the equation above and setting $\delta = \rho^2$ will get our desired result.
\end{proof}

\section{$p$-ary Hidden Matching}
\label{sec:p-ary-HM}
In this section, we analyze a two-party (distributional) one-way communication problem, defined as follows.

\textbf{$p$-ary Hidden Matching problem.} Alice gets a random vector $x \in \mathbb{Z}_p^n$. Bob gets a random $\alpha$-partial matching $G$ (i.e., a matching of size $\alpha n$ on $\{1,2,\dots,n\}$) and a vector $w \in \mathbb{Z}_p^{\alpha n}$. Let $M \in \{0,1\}^{\alpha n \times n}$ be the incidence matrix of $G$, i.e., $M_{ev}=1$ if $v$ is an endpoint of $e$ and $0$ otherwise. There are two choices for the distribution of $w$, distinguishing which is the communication problem.
\begin{itemize}
    \item
In the \textbf{YES} distribution $w$ is correlated with $x$ as $w = Mx$ (arithmetic done in $\mathbb{Z}_p$); 
\item in the \textbf{NO} distribution, $w$ is uniformly random in $\mathbb{Z}_p^n$ (and thus independent of $x$). 
\end{itemize}

Alice must send a message to Bob, based on which Bob needs to distinguish distribution $w$ belongs to. Formally, Bob must output Yes or No (based on Alice's message and his input $w$), and we say a protocol achieves advantage $\epsilon$ if the difference in probability of Bob outputting Yes differs under the Yes and No distributions by at least $\epsilon$. The following shows that Alice needs to send at least $\Omega(\sqrt{n})$ bits for Bob to achieve constant advantage.

\begin{theorem}\label{phm}
For $\alpha \in (0, 1/4]$, any protocol that achieves advantage $\epsilon > 0$ for the $p$-ary Hidden Matching problem requires at least $\Omega(\epsilon \sqrt{n})$ bits of communication from Alice to Bob.
\end{theorem}

The proof of the above lemma is the main result of this section. Our proof closely follows the structure of \cite{gavinsky2007exponential}, from which the main difference is that our proof has to work for the $p$-ary case.

Before we embark on the proof, we need some more lemmas. We begin with an application of hypercontractivity to bound the Fourier mass at any level.

\begin{lemma}
\label{lem:fourier-level}
For a set $A \subseteq \mathbb{Z}_p^n$ with size at least $p^n/2^c$ and let $f$ be its indicator function and let $|z|$ denote the number of non-zero coordinates of $z\in \mathbb{Z}_p^n$. Then for every $k \le 4c$ we have 
$$\frac{p^{2n}}{|A|^2} \sum_{|z|=k} |\hat{f}(z)|^2 \le \left(\frac{4\sqrt{2}p c}{k}\right)^k \ . $$
\end{lemma}
\begin{proof}
By Lemma \ref{lemhct}, given some constant $0 \le \delta \le 1/p$, we have
\begin{align*}
    \frac{p^{2n}}{|A|^2} \sum_{|z|=k} |\hat{f}(z)|^2 &\le
    \frac{p^{2n}}{|A|^2} \frac{1}{\delta^k}\sum_{z\in \mathbb{Z}_p^n} \delta^{|z|}|\hat{f}(z)|^2 \\
    &\le \frac{p^{2n}}{|A|^2} \frac{1}{\delta^k} \left( \frac{|A|}{p^n}\right)^{2/(1+p \delta)} \\
    &= \frac{1}{\delta^k} \left( \frac{p^{n}}{|A|}\right)^{2p\delta/(1+p \delta)} \\
    &\le \frac{1}{\delta^k} \left( \frac{p^{n}}{|A|}\right)^{2p\delta}.
\end{align*}
Choosing $\delta = k/4cp$ will give our desired result.
\end{proof}
We also need a combinatorial lemma about counting of some matchings.
\begin{lemma}
Let $G$ be a uniformly random $\alpha$-partial matching and $M$ be its incidence matrix. If $x \in \mathbb{Z}_p^n$ has $|x|=k$ for some even\footnote{We note that if $|x|$ is odd, then there can be no $z$ such that $M^T z =x$.} $k$, then
$$\Pr_G[\exists z \in \mathbb{Z}_p^{\alpha n} \mathrm{s.t.}M^T z = x] \le \binom{\alpha n}{k/2} \bigg/ \binom{n}{k} \ . $$
\end{lemma}
\begin{proof}
We know that the total number of all $\alpha$-partial matchings of $n$ vertices is $n! / (2^{\alpha n} (\alpha n)! (n-2\alpha n)!)$. And if there exists some $z$ such that $M^T z = x$, then $G$ must have exactly $k/2$ edges between those vertices $v$ with $x_v \ne 0$. There are $k! / (2^{k/2}(k/2)!)$ number of ways to choose those edges. Also, we need to choose $\alpha n - k/2$ edges amongst those $v$ whose $x_v = 0$, which we have $(n-k)! / (2^{n - k} (\alpha n - k/2)! (n-2\alpha n)!)$ ways to do. Combining them together leads to the lemma. 
\end{proof}

From the lemmas above, we can derive an important result in our proof.
\begin{lemma}\label{mainphm}
Let $A \subseteq \mathbb{Z}_p^n$ be of size at least $p^n/2^c$ for some $c \ge 1$, $G$ be a uniformly random $\alpha$-partial matching for some $0 < \alpha \le 1/4$ and $M$ be its incidence matrix. There exists a constant $\gamma$ independent of $n, c$ and $\alpha$, such that for all $\epsilon > 0$, if $c \le \gamma \epsilon\sqrt{n/\alpha}$ then
$$E_M[\norm{p_M-U}_{tvd}] \le \epsilon ,$$
where $p_M(w) = |\{x \in A \mid Mx = w\} / |A| $ is the distribution of $w$ in the \textbf{YES} case when $x$ is uniformly random in $A$.
\end{lemma}

\begin{proof}
To show that $E_M[\norm{p_M-U}_{tvd}] \le \epsilon$, we can start by bounding the Fourier coefficients of $p_M$. In fact they are closely related to $\hat{f}$ (where recall that $f$ is the indicator function for membership in the set $A$):
\begin{align*}
    \widehat{p_M}(z) &= \frac{1}{p^{\alpha n}} \sum_{w \in \mathbb{Z}_p^{\alpha n}} p_M(z) \omega^{-w \cdot z} \nonumber \\
    &= \frac{1}{|A|p^{\alpha n}} \sum_{k=0}^{p-1} \omega^{-k} |\{x \in A | (Mx) \cdot z = k \}| \nonumber \\
    &= \frac{1}{|A|p^{\alpha n}} \sum_{k=0}^{p-1} \omega^{-k} |\{x \in A | x \cdot (M^Tz) = k \}| \nonumber \\
    &= \frac{1}{|A|p^{\alpha n}}\sum_{x \in A} \omega^{-x \cdot (M^T z)} \nonumber \\
    &= \frac{p^n}{|A|p^{\alpha n}} \widehat{f}(M^Tz)
\end{align*}
\noindent
From the bound of Fourier coefficients, we can give a bound on squared total variation distance
\begin{align*}
    E_M[\norm{p_M-U}_{tvd}^2] &\le p^{2 \alpha n} E_M[\norm{p_M-U}_2^2] \\
    &= p^{2\alpha n} E_M \left[ \sum_{z \in \mathbb{Z}_p^{\alpha n}\backslash \{0^{\alpha n}\}} |\widehat{p_M}(z)|^2 \right]\\
    &= \frac{p^{2 n}}{|A|^2} E_M \left[ \sum_{z \in \mathbb{Z}_p^{\alpha n}\backslash \{0^{\alpha n}\}} |\widehat{f}(M^T z)|^2 \right]
    \intertext{by Cauchy-Schwarz inequality, Parseval equality and the bound above. Since there is at most one $z \in \mathbb{Z}_p^{\alpha n}$ such that $x = M^T z$ for given $x$, we have}
    &= \frac{p^{2 n}}{|A|^2} E_M \left[ \sum_{x \in \mathbb{Z}_p^n\backslash \{0^{n}\}} |\widehat{f}(x)|^2 |\{z\in \mathbb{Z}_p^{\alpha n}| x=M^Tz\}| \right] \\
    &= \frac{p^{2 n}}{|A|^2}\sum_{x \in \mathbb{Z}_p^n\backslash \{0^{n}\}} \Pr_M[ \exists z\in \mathbb{Z}_p^{\alpha n} \text{s.t. }M^Tz = x] |\hat{f}(x)|^2 \\
    &\le \frac{p^{2 n}}{|A|^2} \sum_{k=2, k \text{even}}^{2\alpha n} \frac{\binom{\alpha n}{k/2}}{ \binom{n}{k}} \sum_{|x|=k} |\hat{f}(x)|^2.
\end{align*}
We then split the sum into two parts $k < 4c$ and $k \ge 4c$. For $k < 4c$, using $(n/k)^k \le \binom{n}{k} \le (en/k)^k$, we have
\begin{align*}
   \frac{p^{2 n}}{|A|^2}\sum_{k=2, k \text{even}}^{4c-2} \frac{\binom{\alpha n}{k/2}}{ \binom{n}{k}} \sum_{|x|=k} |\hat{f}(x)|^2
   &\le \sum_{k=2, k \text{even}}^{4c-2} \frac{(2e\alpha n / k)^{k/2}}{(n/k)^k}\left(\frac{4\sqrt{2}p c}{k}\right)^k \quad \text{(using Lemma~\ref{lem:fourier-level})} \\
   &\le \sum_{k=2, k \text{even}}^{4c-2} \left(\frac{64e\gamma^2\epsilon^2 p^2}{k}\right)^{k/2}, 
\end{align*}
which is at most $\epsilon^2/2$ when $\gamma$ is sufficiently small. For $k \ge 4c$ note that $\sum_{x} |\hat{f}(x)|^2 = |A|/p^n$ by Parseval and $\binom{\alpha n}{k/2}\big/ \binom{n}{k}$ is decreasing for even $k \le 2\alpha n$, we have
\begin{align*}
\frac{p^{2 n}}{|A|^2}\sum_{k=4c, k \text{even}}^{2\alpha n} \frac{\binom{\alpha n}{k/2}}{ \binom{n}{k}} \sum_{|x|=k} |\hat{f}(x)|^2 & \le 2^c\frac{\binom{\alpha n}{2c}}{ \binom{n}{4c}} \\
& \le 2^c \left( \frac{8c\alpha e}{n} \right)^{2c} \\
& \le \left(8\sqrt{2} e \gamma \epsilon \sqrt{\alpha/n}\right)^{2c}\le \epsilon^2/2.
\end{align*}
The last inequality holds because $n \ge 1$ and $c \ge 1$, and we let $\gamma$ be a sufficiently small constant.
Thus, in total we have $E_M[\norm{p_M-U}_{tvd}^2] \le \epsilon^2$, which means by Jensen $E_M[\norm{p_M-U}_{tvd}]\le \epsilon$.
\end{proof}

From the lemma above, we can prove the communication lower bound of $p$-ary Hidden Matching problem.
\begin{proof}[Proof of Theorem \ref{phm}]
By fixing the randomness of the protocol, we can assume without loss of generality that the protocol is deterministic . Fix $\epsilon > 0$ to a small constant and let $c = \gamma \epsilon \sqrt{n/\alpha}$. Consider any protocol that communicates at most $C = c - \log(1/\epsilon)$ bits. In the protocol, Alice's message gives an partition of $\mathbb{Z}_p^n$ into $2^C$ subsets. We call the sets with size $\epsilon p^n/2^C = p^n/2^c$ be ``large sets'', then for a uniformly random $x \in \mathbb{Z}_p^n$, with probability $1-\epsilon$, $x$ belongs to a large set. When $x$ is in a large set, by Lemma \ref{mainphm}, Bob can get an advantage of at most $\epsilon$. Together with the advantage from small sets, the overall advantage Bob can get is at most $O(\epsilon)$, which completes the proof.
\end{proof}

\section{Reduction to Streaming Algorithm for Unique Games}
In this section, we will prove Theorem~\ref{main}. Towards this end, we will describe a pair of distributions, \textbf{Y} and \textbf{N}, where \textbf{Y} is supported on satisfiable instances of Unique Games, and \textbf{N} is supported with high probability on instances where at most $\approx 1/p$ fraction of constraints can be satisfied. We will then establish, via reduction from the $p$-ary Hidden Matching communication problem, that any low-space streaming algorithm cannot distinguish between these distributions, thus establishing Theorem~\ref{main}.

\subsection{Input distributions}
We construct the above-mentioned distributions in a ``multi-stage'' way (using $k$ stages) based on the \textbf{YES} and \textbf{NO} distributions (defined at the beginning of Section~\ref{sec:p-ary-HM}) for $p$-ary Hidden Matching.  First we independently sample $k$ $\alpha$-partial matchings on $n$ vertices a 
The Unique Games instance graph $G$ will be the union of these matchings. It will thus have $n$ vertices and $k\alpha n$ edges (we allow multiple edges should they be sampled). We next specify the Unique Games constraints, which will be two-variable linear equations, one for each edge.
\begin{itemize}
    \item In the \textbf{Y} distribution, we sample a random $z \in \mathbb{Z}_p^n$ uniformly. We let the constraint on edge $(u,v)$ of $G$ be $x_u + x_v = z_u + z_v$.
\item In the \textbf{N} distribution, for each edge $(u,v)$ of $G$, we let the constraint be $x_u + x_v = q$ for a random $q \in \mathbb{Z}_p$, independently chosen for each edge. 
\end{itemize}


For instances sampled in the \textbf{Y} distribution, the best solution is obviously $x_u = z_u$ for all $u \in [n]$, which satisfies all the constraints. For the \textbf{N} distribution, we can use Chernoff bounds to upper bound the value of the optimal solution.

\begin{lemma}\label{sol}
Let $0 < \epsilon < 1$. If $k = Cp \log p / (\alpha \epsilon^2)$ for some large constant $C > 0$, then for a Unique Games instance sampled from the $\textbf{N}$ distribution, the optimal fraction of satisfiable constraints is at most $(1+\epsilon)/p$ with high probability.
\end{lemma}

Before we proceed to the proof, we first state the Chernoff bound for negatively correlated random variables.

\begin{lemma}[\cite{panconesi1997randomized}]
Let $X_1, \ldots, X_n$ be negatively correlated Bernoulli random variables and $X = X_1 + \cdots + X_n$. Then we have
$$\Pr[X \ge (1+\epsilon)E[X]] \le \exp (-E[X]\epsilon^2/3).$$
\end{lemma}

\begin{proof}[Proof of Lemma \ref{sol}]
Fix an assignment $x\in \mathbb{Z}_p^n$. For $1 \le \ell \le k$, let $X^{(\ell)}_{ij}$ be the indicator of the following event: ``in the $\ell$-th stage, the edge $(i,j)$ is included in the $\alpha$-partial matching and is satisfied by the assignment $x$.'' Then, $S =\sum_{\ell,i,j} X^{(\ell)}_{ij}$, summed over $1 \le \ell k$, and $1 \le i < j \le n$, is the random variable counting the number of constraints satisfied by the assignment $x$. Note that $\mathbb{E}[S] = k\alpha n/p$ is the expected number of constraints by the assignment $x$. And we know that each $X^{(\ell)}_{ij}$ is a Bernoulli random variable with probability of equaling $1$ being $2\alpha n / (pn(n-1))$.

We first claim that these random variables are negatively correlated. In fact, edges in different stages are independent. For edges in the same stage $\ell$, consider that we know that random variables $X^{(\ell)}_{i_1 j_1}, X^{(\ell)}_{i_2 j_2}, \ldots, X^{(\ell)}_{i_t j_t}$ have value $1$, and a vertex pair $(i_0, j_0)$. If $i_0$ or $j_0$ is occurred in some $i_s$ or $j_s$, then $X^{(\ell)}_{i_0j_0}$ must be $0$. Otherwise the conditional expectation of $X^{(\ell)}_{i_0j_0}$ is $2(\alpha n - t) / (p(n-t)(n-1-t))$, which is less than the unconditional expectation of $2\alpha n / (pn(n-1))$. In all cases we have $E[X^{(\ell)}_{i_0j_0} \mid ~ X^{(\ell)}_{i_1 j_1}= X^{(\ell)}_{i_2 j_2}= \cdots= X^{(\ell)}_{i_t j_t}=1] \le E[X^{(\ell)}_{i_0j_0}]$, which in turn means negative correlation.

Thus, by Chernoff bound for negatively random variables, we know that
$$\Pr[S \ge (1+\epsilon)k\alpha n / p] \le \exp (-\epsilon^2k\alpha n /3 p) = p^{-Cn/3} \le p^{-2n}. $$
The proof is now complete by a union bound over all $p^n$ candidate assignments.
\end{proof}

\subsection{Reduction from $p$-ary Hidden Matching}

Note that each stage of constraints in the Unique Games instance corresponds to the $p$-ary Hidden Matching problem, with the $\textbf{Y}$ distribution (resp. \textbf{N} distribution) coinciding with the YES distribution (NO distribution) of the Hidden Matching problem.
Using this, we can link the hardness of the two problems via a hybrid argument. Recall that we say that a decision algorithm distinguishes between two distributions $D_1$ and $D_2$ with advantage $\eta$ if it accepts samples from one distribution with probability at least $\eta$ more than those from the other distribution. 

\begin{lemma}\label{red}
Suppose there exists a streaming algorithm \textbf{ALG} using $c$ bits of memory that can achieve advantage $1/4$ in distinguishing between instances from the \textbf{Y} and \textbf{N} distributions of Unique Games instances. Then there exists a protocol with $c$ bits of communication for the $p$-ary Hidden matching problem with advantage $\Omega(1/k)$ in distinguishing between \textbf{YES} and \textbf{NO} distributions.
\end{lemma}

We now prepare for the proof of Lemma~\ref{red}. Our proof follows along the lines of a similar argument in \cite{kapralov2015streaming}. In the execution of \textbf{ALG} on instances from the \textbf{Y} and \textbf{N} distributions, let the memory after receiving the $i$-th stage constraints be $S^Y_i$ and $S^N_i$ respectively. Thus $S^Y_i, S^N_i$ are random variables in $\{0,1\}^c$. Without loss of generality, we assume that $S^Y_0=S^N_0=0$.

We now define the notion of an informative index, as in \cite{kapralov2015streaming}.
\begin{definition}[Informative index]
An index $j \in \{0, \ldots, k-1\}$ is said to be $\delta$-informative for $\delta > 0$ if $$\norm{S^Y_{j+1}-S^N_{j+1}}_{tvd} \ge \norm{S^Y_{j}-S^N_{j}}_{tvd} + \delta$$
\end{definition}
We now show the existence of a $\Omega(1/k)$-informative index for any streaming algorithm that distinguishes between $\textbf{Y}$ and $\textbf{N}$ distributions. 
\begin{lemma}
Suppose a streaming algorithm \textbf{ALG} uses $c$ bits of memory and distinguishes the \textbf{Y} and \textbf{N} distributions with advantage $1/4$. Then the algorithm has a $\Omega(1/k)$-informative index. 
\end{lemma}
\begin{proof}
At first, $\norm{S^Y_{0}-S^N_{0}}_{tvd} = 0$; at the end of the algorithm, since advantage is at least $1/4$, $\norm{S^Y_{k}-S^N_{k}}_{tvd}$ must be at least some constant $C$. Let $j$ be the first index such that $\norm{S^Y_{j+1}-S^N_{j+1}}_{tvd} \ge C(j+1)/k$, then $j$ is a $C/k$-informative index.
\end{proof}

Let $j^*$ be a $\Omega(1/k)$-informative index of a streaming algorithm \textbf{ALG}. Using \textbf{ALG}, we can devise a communication protocol for the $p$-ary Hidden Matching problem as follows.

\begin{enumerate}
\itemsep=0.5ex
    \item Suppose Alice holds as input a random string $x \in \mathbb{Z}_p^n$. She samples $j^*$ random $\alpha$-partial matchings and feeds the streaming algorithm UG constraints for the first $j^*$ stages that follow the \textbf{Y} distribution with the setting $z = x$.
    
    \item Alice sends the memory contents of \textbf{ALG} after $j^*$ stages to Bob.
    
    \item Bob samples an $\alpha$-partial matching and gives constraints $x_u+x_v = w_e$ for $e=(u,v)$ according to his $w$. He then continues running \textbf{ALG} on these constraints as the $(j^*+1)$'th stage.
    
    \smallskip
    Let the memory Bob gets be $s$.
    
    \item 
        Let the resulting memory distribution under the two cases (depending on $w$'s distribution) be $\tilde{S}^{\text{YES}}$ and $\tilde{S}^{\text{NO}}$. (Note that these distribution can be computed by Bob since \textbf{ALG} is known.)

    \smallskip
    Bob outputs $1$ if $\Pr[\tilde{S}^{\text{YES}} = s] \ge \Pr[\tilde{S}^\text{NO} = s]$, and otherwise $0$.
\end{enumerate}

The above completes the description of the reduction. 
Before we analyze it and proceed to the proof of Lemma~\ref{red}, we need the following fact about the statistical (total variation) distance between random variabls.

\begin{lemma}[Claim 6.5, \cite{kapralov2015streaming}]\label{ftvd}
Let $X, Y$ be two random variables and $W$ be independent of $(X,Y)$. Then for any function $f$, we have
$$\norm{f(X,W) - f(Y,W)}_{tvd} \le \norm{X-Y}_{tvd} \ . $$
\end{lemma}

\smallskip
\begin{proof}[Proof of Lemma~\ref{red}]
We argue that the above protocol for $p$-ary Hidden Matching achieves the claimed advantage of $\Omega(1/k)$ in distinguishing between \textbf{YES} and \textbf{NO} distributions.

Let $f$ be the function that maps the memory after stage $j^*$ and constraints of stage $(j^*+1)$ to the memory after stage $(j^*+1)$. Thus we have $\tilde{S}^{\text{YES}}= S^Y_{j^*+1} = f(S^Y_j, C^Y)$ and $\tilde{S}^{\text{NO}} = f(S^Y_j, C^N)$, where $C^Y, C^N$ be the constraints Bob generated in both cases. We also know that $S^N_{j^*+1} = f(S^N_j, C^N)$.

By Lemma \ref{ftvd}, we know that 
$$\norm{\tilde{S}^{\text{NO}} - S^N_{j^*+1}}_{tvd} = \norm{f(S^Y_{j^*}, C^N)-f(S^N_{j^*}, C^N)}_{tvd} \le \norm{S^Y_{j^*}-S^N_{j^*}}_{tvd}. $$
Hence, we have
\begin{align*}
\norm{\tilde{S}^{\text{YES}}-\tilde{S}^{\text{NO}}}_{tvd} &\ge \norm{S^Y_{j^*+1}-S^N_{j^*+1}}_{tvd} -  \norm{\tilde{S}^{\text{NO}} - S^N_{j^*+1}}_{tvd}\\
&\ge \norm{S^Y_{j^*+1}-S^N_{j^*+1}}_{tvd} - \norm{S^Y_{j^*}-S^N_{j^*}}_{tvd} \\
&\ge \Omega(1/k).
\end{align*}
The strategy in Step 4 that Bob uses distinguishes between $\tilde{S}^{\text{YES}}$ and $\tilde{S}^{\text{NO}}$ with advantage exactly $\norm{\tilde{S}^{\text{YES}}-\tilde{S}^{\text{NO}}}_{tvd}$, which is at least $\Omega(1/k)$. This concludes the proof of Lemma~\ref{red}.
\end{proof}

\smallskip
Our main result, Theorem~\ref{main}, now follows by choosing $\alpha = 1/8$ and $k = \lceil C p \log p/\epsilon^2 \rceil$  for a large enough absolute constant $C$, and combining together Theorem~\ref{phm}, Lemma~\ref{red}, and Lemma~\ref{sol}.

\section{Conclusion}

We proved that Unique Games is hard for single-pass streaming algorithms in a strong sense: even if the instance is perfectly satisfiable, the algorithm cannot certify that it is even $(1/p+\epsilon)$-satisfiable, where $p$ is the alphabet size, and $\epsilon > 0$ is an arbitrary constant. Some natural directions to extend our lower bound would be to multi-pass algorithms, and for random arrival order of the constraints. 

An interesting direction for future work would be to establish limitations of streaming algorithms for other approximation problems which are only known to be ``Unique Games-hard.'' An example, which partly motivated this work initially, is the Maximum Acyclic Subgraph (MAS) problem. The MAS problem is another one of those notorious problems for which there is a trivial algorithm that achieves approximation ratio of $2$ (the algorithm is simply to order the vertices arbitrarily, and take either all the forward-going or backward-going edges as an acyclic subgraph with at least $1/2$ the edges), and no efficient algorithm achieving a factor $(2-\epsilon)$-approximation is known for any fixed $\epsilon > 0$. On the other hand, known NP-hardness results are rather weak, but under the Unique Games conjecture, it is known that there is no efficient $(2-\epsilon)$-approximation for MAS~\cite{GMR,GHMRC}.

One can try to explain the difficulty of MAS in the streaming model, by proving a result similar in spirit to the result we established for Unique Games. Specifically, given as input a directed graph whose edges arrive one-by-one, can a low-space single-pass streaming algorithm distinguish between the cases when the directed graph is acyclic and when it has no acyclic subgraph with even $1/2+\epsilon$ of the edges? (The $1/2$ threshold being trivial, since any directed graph has an acyclic subgraph with $1/2$ the edges.) A result of this flavor was shown with $1/2$ replaced by $7/8$ in \cite{GVV-approx17}. 

The reduction from Unique Games to $(2-\epsilon)$-approximating MAS~\cite{GMR} and our inapproximability result for UG in the streaming model gives hope to prove the desired streaming hardness for MAS as well, by implementing the reduction in a streaming manner. Since reductions involving CSPs are usually local, the arrival of one constraint of problem $\mathbb{A}$ can be mimicked by the arrival of the constraints of problem $\mathbb{B}$ that implement it. The reduction from UG to MAS (and indeed many other CSPs), however, introduces constraints between all pairs of variables that share a constraint with a UG vertex $u$. So to implement it one would need the UG streaming hardness under a ``vertex arrival'' model, where the graph is bipartite, and all constraints involving a left hand side vertex arrive in sequence. We \emph{can} adapt the reduction in \cite{GMR} to something local, based only on a single constraint, thereby making it more friendly to the edge arrival model. However, this only yields a weaker hardness result that distinguishing DAGs from graphs whose MAS has at most $\approx 3/4$ edges requires $\Omega(\sqrt{n})$ space. 

Obtaining a tight streaming hardness result for MAS, and more broadly leveraging our tight streaming hardness result for Unique Games toward streaming inapproximability results for other optimization problems for which we have optimal reductions from Unique Games, are interesting directions for future work. Further, given the hardness results in this work and \cite{kapralov2015streaming}, one can ask which CSPs and related problems admit non-trivial approximate estimation algorithms in the streaming model. Even though one might suspect that strong hardness results should be pervasive, it seems that it is rather non-trivial to establish strong limitations of streaming algorithms, and the algorithms for Max 2CSP in \cite{GVV-approx17} suggest that there might be more interesting cases where streaming algorithms can provide non-trivial guarantees.

\bibliography{main}

\begin{thebibliography}{10}

\bibitem{CMM}
Moses Charikar, Konstantin Makarychev, and Yury Makarychev.
\newblock Near-optimal algorithms for unique games.
\newblock In {\em Proceedings of the 38th Annual {ACM} Symposium on Theory of
  Computing}, pages 205--214, 2006.

\bibitem{DKKMS-stoc18}
Irit Dinur, Subhash Khot, Guy Kindler, Dor Minzer, and Muli Safra.
\newblock Towards a proof of the 2-to-1 games conjecture?
\newblock In {\em Proceedings of the 50th Annual {ACM} Symposium on Theory of
  Computing}, pages 376--389, 2018.

\bibitem{feige-reichman}
Uriel Feige and Daniel Reichman.
\newblock On systems of linear equations with two variables per equation.
\newblock In {\em Approximation, Randomization, and Combinatorial Optimization,
  Algorithms and Techniques (APPROX, RANDOM)}, pages 117--127, 2004.

\bibitem{gavinsky2007exponential}
Dmitry Gavinsky, Julia Kempe, Iordanis Kerenidis, Ran Raz, and Ronald De~Wolf.
\newblock Exponential separations for one-way quantum communication complexity,
  with applications to cryptography.
\newblock In {\em Proceedings of the 39th annual ACM symposium on Theory of
  computing}, pages 516--525. ACM, 2007.

\bibitem{GHMRC}
Venkatesan Guruswami, Johan H{\aa}stad, Rajsekar Manokaran, Prasad Raghavendra,
  and Moses Charikar.
\newblock Beating the random ordering is hard: Every ordering {CSP} is
  approximation resistant.
\newblock {\em SIAM Journal on Computing}, 40(3):878--914, 2011.

\bibitem{GMR}
Venkatesan Guruswami, Rajsekar Manokaran, and Prasad Raghavendra.
\newblock Beating the random ordering is hard: Inapproximability of maximum
  acyclic subgraph.
\newblock In {\em Proceedings of the 49th Annual {IEEE} Symposium on
  Foundations of Computer Science}, pages 573--582, 2008.

\bibitem{GVV-approx17}
Venkatesan Guruswami, Ameya Velingker, and Santhoshini Velusamy.
\newblock Streaming complexity of approximating max 2csp and max acyclic
  subgraph.
\newblock In {\em 20th International Workshop on Approximation Algorithms for
  Combinatorial Optimization Problems (APPROX)}, pages 8:1--8:19, 2017.

\bibitem{Kahn1988The}
Jeff Kahn, Gil Kalai, and Nathan Linial.
\newblock The influence of variables on boolean functions.
\newblock In {\em Proceedings of the 29th Annual Symposium on Foundations of
  Computer Science}, pages 68--80, 1988.

\bibitem{kapralov2015streaming}
Michael Kapralov, Sanjeev Khanna, and Madhu Sudan.
\newblock Streaming lower bounds for approximating max-cut.
\newblock In {\em Proceedings of the 26 annual ACM-SIAM symposium on Discrete
  Algorithms}, pages 1263--1282. Society for Industrial and Applied
  Mathematics, 2015.

\bibitem{Khot-UGC}
Subhash Khot.
\newblock On the power of unique 2-prover 1-round games.
\newblock In {\em Proceedings on 34th Annual {ACM} Symposium on Theory of
  Computing}, pages 767--775, 2002.

\bibitem{KKMO}
Subhash Khot, Guy Kindler, Elchanan Mossel, and Ryan O'Donnell.
\newblock Optimal inapproximability results for {MAX-CUT} and other 2-variable
  {CSP}s?
\newblock {\em SIAM Journal on Computing}, 37(1):319--357, 2007.

\bibitem{KMS-focs18}
Subhash Khot, Dor Minzer, and Muli Safra.
\newblock Pseudorandom sets in grassmann graph have near-perfect expansion.
\newblock {\em Electronic Colloquium on Computational Complexity {(ECCC)}},
  25:6, 2018.
\newblock To appear in FOCS'18.

\bibitem{o2014analysis}
Ryan O'Donnell.
\newblock {\em Analysis of boolean functions}.
\newblock Cambridge University Press, 2014.

\bibitem{panconesi1997randomized}
Alessandro Panconesi and Aravind Srinivasan.
\newblock Randomized distributed edge coloring via an extension of the
  chernoff--hoeffding bounds.
\newblock {\em SIAM Journal on Computing}, 26(2):350--368, 1997.

\bibitem{Verbin2011The}
Elad Verbin and Wei Yu.
\newblock The streaming complexity of cycle counting, sorting by reversals, and
  other problems.
\newblock In {\em Proceedings of the 26 annual ACM-SIAM symposium on Discrete
  Algorithms}, pages 11--25, 2011.

\end{thebibliography}

\end{document}